\documentclass[11pt,a4paper]{article}

\usepackage{jheppub}
\usepackage[T1]{fontenc}
\usepackage[utf8]{inputenc}
\usepackage{amsthm,mathtools}
\usepackage[nameinlink,noabbrev]{cleveref}
\usepackage{microtype}

\allowdisplaybreaks
\newcommand{\AdS}{\mathrm{AdS}}

\newcommand{\End}{\operatorname{End}}

\newcommand{\Span}{\operatorname{span}}
\newcommand{\Tr}{\operatorname{Tr}}
\newcommand{\cB}{\mathcal B}
\newcommand{\cD}{\mathcal D}
\newcommand{\cI}{\mathcal I}
\newcommand{\cP}{\mathcal P}
\newcommand{\cM}{\mathcal M}
\newcommand{\cS}{\mathcal S}
\newcommand{\cT}{\mathcal T}
\newcommand{\cU}{\mathcal U}
\newcommand{\PBW}{\mathrm{PBW}}
\newcommand{\FF}{\mathrm{F}}
\newcommand{\doi}[1]{\href{https://doi.org/#1}{\nolinkurl{doi:#1}}}
\newcommand{\ket}[1]{\left|#1\right\rangle}
\newcommand{\bra}[1]{\left\langle#1\right|}
\newcommand{\braket}[2]{\left\langle#1\middle|#2\right\rangle}

\newtheorem{theorem}{Theorem}[section]
\newtheorem{lemma}[theorem]{Lemma}
\newtheorem{corollary}[theorem]{Corollary}
\newtheorem{proposition}[theorem]{Proposition}
\theoremstyle{definition}
\newtheorem{definition}[theorem]{Definition}
\crefname{theorem}{theorem}{theorems}
\Crefname{theorem}{Theorem}{Theorems}
\crefname{lemma}{lemma}{lemmas}
\Crefname{lemma}{Lemma}{Lemmas}
\crefname{corollary}{corollary}{corollaries}
\Crefname{corollary}{Corollary}{Corollaries}
\crefname{proposition}{proposition}{propositions}
\Crefname{proposition}{Proposition}{Propositions}

\title{Finite stress-tensor moment tomography of boundary-graviton coherences}

\author{Dyuman Bhattacharya}
\author{and Robert B. Mann}
\affiliation{Department of Physics and Astronomy, University of Waterloo,\\
Waterloo, Ontario N2L 3G1, Canada}

\emailAdd{d7bhatta@uwaterloo.ca}
\emailAdd{rbmann@uwaterloo.ca}

\abstract{
Einstein gravity in 2+1 dimensions is a topological theory that does not admit the existence of local bulk gravitons. However, \(\AdS_3\) gravity with Brown--Henneaux boundary conditions does admit the existence of boundary gravitons, which are described by Virasoro descendants. An observer living on the asymptotic boundary can use finite order moments of the boundary stress tensor as a probe of these states.  We fix one highest-weight module and assume that the Shapovalov forms are nondegenerate.  We also assume that the unknown state the observer is probing has a known finite descendant support and that interlevel coherence is part of the unknown data.  We find that, given a descendant at level \(N\geq2\), stress moment tomography at order \(N\) (meaning \(N\) stress moment insertions) is informationally incomplete, while tomography with stress moments of order \(N+1\) is informationally complete.  This is to say that, on the unitary regular locus, orthogonal pure states can have the same complete order-\(N\) record.  Likewise, for exact finite-window result for levels through \(L\geq2\), tomography of order \(L\) is incomplete while tomography of order \(L+1\) recovers every density-matrix block.  Corresponding exact thresholds hold in the full left--right theory.  In heavy sectors for which a semiclassical bulk interpretation exists, these thresholds say when finite-support boundary-graviton states first become distinguishable.
}

\begin{document}

\maketitle
\flushbottom

\section{Introduction}
\label{sec:intro}

A general expectation of a quantum theory of gravity is that it should be possible to place spacetime in a state of superposition. Such a procedure in one sense is straightfoward enough -- for example one could superpose two different solutions to the Wheeler-DeWitt equation. In so doing, however, one is left with the issue of detection: what is the response of a given probe if placed in a region where a given mass is in spatial superposition, or alternatively the spacetime is in a superposition of two energy eigenstates? 
Recent research has addressed this question in its simplest setting, namely the superposition of a Banados-Teitelboim-Zanelli (BTZ) black hole \cite{Banados:1992wn,Banados:1992gq} in two different states
\cite{Foo:2021exb,Suryaatmadja:2023onb}.  The response of the detector 
was shown to exhibit a wavelike pattern 
as a function of the ratio of the horizons in the two branches of the superposition, 
punctuated by spikes at rational value of this ratio. Such investigations have been extended to superpositions of flat spacetime with different topologies
\cite{Foo:2022dnz,Goel:2024vtr}.

More generally, how might one distinguish a genuine superposition of two spacetime states from a classical mixture of those states? We consider here this problem from a holographic perspective, asking how a rigorous probe of superposed spacetimes could be obtained via boundary observables. The simplest non-trivial setting in which to address this problem is 
three-dimensional Einstein gravity. This theory has no local propagating graviton, but its phase space is not empty.  Under asymptotically $\AdS_3$ boundary conditions there remain two Brown--Henneaux Virasoro actions, each with its classical central extension \cite{Brown:1986nw}, and the associated large diffeomorphisms generate the boundary gravitons.

Consider an observer located in the $2$-dimensional asymptotic boundary of $\AdS_3$ who is given one of two density matrices: either a coherent superposition of two boundary-graviton  states, or the corresponding incoherent mixture with the same diagonal probabilities. The only difference between the two states is an off-diagonal phase. The observer is able to measure boundary stress-tensor data, but  only up to some finite moment order (one point data, two point stress tensor products, three point products, etc.).
This is the operational/limiting aspect of the problem:   how much finite stress tensor data is needed before this observer can detect the phase? Or, equivalently, which boundary-graviton coherences are invisible to all stress-tensor measurements below a given finite order?

We first ask whether bounded-order moments distinguish states on one descendant eigenspace.  The broader problem keeps the highest-weight sector fixed but assumes only bounded descendant support, so coherence between unequal excitation energies must also be recovered.

For a rotating BTZ black hole, locally $\AdS_3$ but globally nontrivial \cite{Banados:1992wn,Banados:1992gq}, the heavy-primary weights encode the reference mass and angular momentum according to
\begin{equation}
 h-\frac c{24}\sim \frac12(\ell M_0+J_0),\qquad
 \bar h-\frac{\bar c}{24}\sim \frac12(\ell M_0-J_0).
\label{eq:btz-dict}
\end{equation}
In a holographic CFT whose chosen heavy sector admits a semiclassical Einstein-gravity description, fixing a primary with weights $(h,\bar h)$ fixes a reference heavy coadjoint orbit, while Virasoro descent supplies its boundary-graviton excitations \cite{Strominger:1997eq,Garbarz:2014kaa,Cotler:2018zff}.  Virasoro descent shifts the conserved charges: on $H_N\otimes\bar H_{\bar N}$ one has $L_0=(h+N)I$ and $\bar L_0=(\bar h+\bar N)I$, giving
\begin{equation}
 h+N-\frac c{24}\sim \frac12(\ell M_{N,\bar N}+J_{N,\bar N}),\qquad
 \bar h+\bar N-\frac{\bar c}{24}\sim \frac12(\ell M_{N,\bar N}-J_{N,\bar N}).
\label{eq:btz-descendant-charges}
\end{equation}
Fixing $(N,\bar N)$ fixes the ADM mass and angular momentum.  We write $H_N=V_{c,h}[N]$ for level $N$ of the chiral Verma module and $P_N$ for its level projection (equivalently, in a unitary representation, the $L_0$ spectral projection).  Nondegeneracy through the cutoff \cite{Kac:1979,KacRaina:1987} makes these level spaces agree with those of the irreducible quotient.  For $N,\bar N\ge2$, \cref{thm:saturation,cor:nonchiral-saturation} give separate saturation orders $N+1$ and $\bar N+1$.

The holomorphic stress tensor has the mode expansion
\begin{equation}
 T(z)=\sum_{n\in\mathbb Z}L_nz^{-n-2}.
\label{eq:Texpansion}
\end{equation}
Nested radial contours preserve the product order and extract its modes from \eqref{eq:Texpansion}.  For repeated preparations of $\rho$ supported on $H_N$, the ideal record through order $k$ is
\begin{equation}
\cM_k(\rho)=
\left(
\Tr_{H_N}\!\left[
\rho P_NL_{n_1}\cdots L_{n_r}P_N
\right]
\right)_{\substack{0\le r\le k\\ \sum_i n_i=0}}
\end{equation}
and a product $\prod_k L_{n_k}$ will be called a word.  The observables with at most $k$ insertions lie in
\begin{equation}
 A_k(N)=P_N\Span\left\{L_{n_1}\cdots L_{n_r}:0\le r\le k,\ \sum_{i=1}^{r}n_i=0\right\}P_N .
\label{eq:Ak-intro}
\end{equation}
Here $r$ counts insertions and the empty word is $I$. A word of grade $q = \sum i n_i$ maps $H_N$ to $H_{N-q}$. The two-sided compression by $P_N$ vanishes unless $q = 0$.

Two states have the same record iff their difference lies in $I_k(N)$, the trace annihilator of $A_k(N)$: $I_k(N) = \{ X : \mathrm{Tr}(XO)=0 \text{ for all } O \in A_k(N) \}$. Therefore $A_k(N) = \mathrm{End}(H_N)$ iff $I_k(N) = \{0\}$ iff $M_k$ is injective.

Thus, it is of some interest to determine the least order on $H_N$ for the projected system $A_k(N)$ and to characterize the structure of the lost information one order below.
To get a handle on the structure of the projected system $A_k(N)$, together with structure hidden just one order below grade $k$, then, we must remove the zero-grade restriction on words at finite level. For a state whose support is known only to lie in $H_{\le L}$, use
\begin{equation}
 H_{\le L}=\bigoplus_{N=0}^{L}H_N,
 \qquad
 Q_L=\sum_{N=0}^{L}P_N,
\end{equation}
and take as the measurement space the compression of all words $W=L_{n_1}\cdots L_{n_r}$, without a condition on their total grade:
\begin{equation}
 A_k^{\le L}
 =Q_L\Span\left\{L_{n_1}\cdots L_{n_r}:0\le r\le k\right\}Q_L
 \subseteq\End(H_{\le L}).
\label{eq:A-window-intro}
\end{equation}
A grade-$q$ word sends $H_s$ to $H_{s-q}$, so $q\ne0$ probes interlevel coherence.  On the regular locus,
\begin{equation}
 A_L^{\le L}\subsetneq\End(H_{\le L}),
 \qquad
 A_{L+1}^{\le L}=\End(H_{\le L})
 \qquad (L\ge2).
\label{eq:window-result-intro}
\end{equation}

At sharp level, choose orthonormal $\ket u,\ket v\in H_N$:
\begin{equation}
 \ket{\psi_\theta}=\frac1{\sqrt2}\left(\ket u+e^{i\theta}\ket v\right),\qquad \rho_\theta=\ket{\psi_\theta}\bra{\psi_\theta}, \qquad
 \rho_{\rm mix}=\frac12\ket u\bra u+\frac12\ket v\bra v.
\end{equation}
The state $\rho_\theta$ differs from the equal incoherent mixture by
\begin{equation}
 \rho_\theta-\rho_{\rm mix}
 =\frac12e^{-i\theta}\ket u\bra v+\frac12e^{i\theta}\ket v\bra u.
\label{eq:coherence-diff}
\end{equation}
By \eqref{eq:coherence-diff}, the two records agree through order $k$ iff $\rho_\theta-\rho_{\rm mix}\in I_k(N)$.

This is not quantum tomography as conventionally done: for the correlator to be a density matrix suffices to know its trace and the joint spectrum of a commuting family of quantum-KdV charges of each charge grade; we make a different use of the observed data, namely the full admissible word span and its trace annihilator.

The failure of a Burnside shortcut \cite{Burnside:1905} to compute $\dim I_k(N)$ is no coincidence: Burnside's theorem applies to irreducible unital algebras, while $A_k(N)$ is not unital nor irreducible: $A_k(N)A_\ell(N)\subseteq A_{k+\ell}(N)$, and a single $A_k(N)$ is not generally closed. Thus $H_N$ is not a Virasoro submodule; as a side issue, enveloping-algebra generation forgets word length and cannot determine this filtration: independent of the purification of $\rho_{\rm mix}$ into a wave functional $h$, annihilators of $h$ will always contain $H_N$.

Jacobson density yields only an unspecified saturation order $k$. Let $L_{c,h}$ be the irreducible quotient, with nondegenerate Shapovalov forms through level $N$: its commutant is $\mathbb C$, so any intertwiner is scalar on the highest-weight line and hence everywhere by cyclicity. Given $L_{c,h}$, density therefore realizes any $T\in\End(H_N)$ on a basis by some $u\in U(\mathrm{Vir})$; decomposing $u$ by $\operatorname{ad}L_0$, the compression $P_NuP_N$ retains a finite degree-zero sum and thus belongs to $A_k(N)$ for some $k$. Thus density yields $A_k(N)=\End(H_N)$ for some finite $k$, but neither the order-$N$ defect, the sufficiency of $N+1$, nor the estimate
\begin{equation}
 %\dim\Ann_{\Tr}(A_N(N))
 \dim I_k(N)
 \ge\left\lfloor\frac N2\right\rfloor.
\end{equation}
All three depend on the word-length filtration discarded by abstract generation.

Ordinary tomography permits designed informationally complete frames \cite{DAriano:2003,DAriano:2004,ParisRehacek:2004}, including under restrictions on states or observables \cite{Heinosaari:2013}.  Here the stress tensor fixes the probes; without controlled evolution or dynamical Lie closure, this is not quantum controllability or observability \cite{Albertini:2003,DAlessandro:2003}.  The resource counted by $k$ is insertion number alone.

For $N\ge2$ on the Shapovalov-regular locus,
\begin{equation}
 %\dim\Ann_{\Tr} (A_N(N))
\dim I_k(N)\ge\left\lfloor\frac N2\right\rfloor,
 \qquad
 A_{N+1}(N)=\End(H_N),
\end{equation}
so $k_{\rm sat}(N)=N+1$.  At a unitary regular point, the order-$N$ annihilator contains a real subspace of dimension at least $\lfloor N/2\rfloor$ whose nonzero elements are Hermitian of rank two with opposite-sign eigenvalues.  Each yields two orthogonal pure states,  which the order-$N$ record identifies with one another and with their equal mixture.  For $L\ge2$, passage to $H_{\le L}$ changes no threshold: the top corner retains the order-$L$ defect, while order $L+1$ fills every block.  With one budget for two chiralities, $N,\bar N\ge2$ and regularity in both sectors give saturation at $N+\bar N+2$.

The equality $A_k(N)=\End(H_N)$ means algebraic informational completeness \cite{DAriano:2003,DAriano:2004}, not a bound on settings, samples, inversion cost, conditioning, or angular resolution.  Hermitian operators on $H_N$ form a real vector space of dimension $p(N)^2$, so an informationally complete Hermitian basis has $p(N)^2$ elements and spans the identity direction.  For the window, write
\begin{equation}
 D_L=\sum_{n=0}^{L}p(n),
\end{equation}
and the corresponding basis has $D_L^2$ elements.

\section{Boundary observer and finite stress-tensor data}
\label{sec:physics}

Brown--Henneaux boundary conditions leave two Virasoro actions on the asymptotic phase space \cite{Witten:1988hc,Witten:1987ty,Alekseev:1988ce,Garbarz:2014kaa,Cotler:2018zff}.  We use one non-vacuum highest-weight module; choosing it heavy permits, when available, a semiclassical interpretation of its descendants as boundary gravitons.  Vacuum and Kac-degenerate modules, which require null-vector quotients, are excluded.  Stress-tensor correlators supply the mode-word data.  In plane radial quantization \cite{DiFrancesco:1997},
\begin{equation}
 T(z)=\sum_{n\in\mathbb Z}L_nz^{-n-2},
 \qquad
 L_n=\frac{1}{2\pi i}\oint dz\,z^{n+1}T(z).
\label{eq:mode-extraction}
\end{equation}
For $\rho$ of finite descendant support, known radially ordered correlators determine the mode moments.  The plane--cylinder map, including its Schwarzian term, and analytic continuation give the same coefficients from Lorentzian cylinder correlators:
\begin{equation}
 \Tr\!\left(\rho L_{n_1}\cdots L_{n_r}\right)
 =\frac{1}{(2\pi i)^r}
 \oint_{|z_1|>\cdots>|z_r|}
 \prod_{j=1}^r dz_j\,z_j^{n_j+1}
 \Tr\!\left(\rho\,T(z_1)\cdots T(z_r)\right).
\label{eq:nested-contours}
\end{equation}
The outer contour extracts $L_{n_1}$, so nesting fixes the displayed order.  Exchanging contours crosses the singular part of the stress-tensor product:
\begin{equation}
 T(z)T(w)\sim \frac{c/2}{(z-w)^4}
 +\frac{2T(w)}{(z-w)^2}
 +\frac{\partial T(w)}{z-w}.
\label{eq:TT-OPE}
\end{equation}
Its residues give the Virasoro commutator, including the central term; normal ordering is the same rearrangement algebraically.  For $W=L_{n_1}\cdots L_{n_r}$ and $q=\sum_jn_j$, one has $[L_0,W]=-qW$ and $W:H_s\to H_{s-q}$.  Hence sharp-level compression keeps only $q=0$, and the record is complete exactly when $A_k(N)$ separates states.  The window imposes no zero-grade restriction.  For a grade-$q$ word $W$,
\begin{equation}
 \Tr(\rho W)
 =\sum_{\substack{0\le s\le L\\0\le s-q\le L}}
 \Tr_{H_s}\!\left(P_s\rho P_{s-q}WP_s\right).
\label{eq:window-block-trace}
\end{equation}
Thus $q\ne0$ probes interlevel blocks.  The window record is
\begin{equation}
 \cM_k^{\le L}(\rho)
 =\left(
 \Tr_{H_{\le L}}\!\left(\rho Q_LL_{n_1}\cdots L_{n_r}Q_L\right)
 \right)_{\substack{0\le r\le k\\ n_i\in\mathbb Z}},
\label{eq:Mk-window}
\end{equation}
and two states share it iff
\begin{equation}
 \cM_k^{\le L}(\rho)=\cM_k^{\le L}(\sigma)
 \quad\Longleftrightarrow\quad
 \Tr_{H_{\le L}}\!\left((\rho-\sigma)O\right)=0
 \quad\text{for all }O\in A_k^{\le L}.
\label{eq:window-annihilator-equivalence}
\end{equation}
Here $Q_L$ imposes the support promise; the individual $P_s$ are not observables.

For non-Hermitian $W$, use
\begin{equation}
 W^\dagger=L_{-n_r}\cdots L_{-n_1},
 \qquad
 W_{\rm R}=\frac{W+W^\dagger}{2},
 \qquad
 W_{\rm I}=\frac{W-W^\dagger}{2i}.
\label{eq:hermitianization}
\end{equation}
For a density operator, these are the real and imaginary parts of the word's complex expectation:
\[
 \Tr(\rho W_{\rm R})=\Re\Tr(\rho W),
 \qquad
 \Tr(\rho W_{\rm I})=\Im\Tr(\rho W).
\]
Adjunction reverses the word and changes grade $q=\sum_i n_i$ to $-q$, but not its length.  Hermitianization therefore adds no insertion.  At sharp level compression removes $q\ne0$; for $q=0$, $P_NWP_N$ and its adjoint lie in $A_k(N)$.  On a window, $Q_LW^\dagger Q_L=(Q_LWQ_L)^\dagger$, and the $q$ and $-q$ pieces are the real and imaginary parts of one functional.  A word need not isolate a matrix element.  A grade-$q$ component evolves as
\begin{equation}
 e^{\tau L_0}W_qe^{-\tau L_0}=e^{-q\tau}W_q
\end{equation}
At $\tau=it$ this is $e^{-iqt}$, or $e^{-iqt/\ell}$ with dimensions restored, while rotation gives $e^{-iq\phi}$.  Nonzero-grade tomography requires common time and angular origins: averaging either reference removes those grades.  One grade-$q$ expectation can vanish despite coherence; only the full dual-band span forces the corresponding blocks to vanish, and saturation supplies that span.

Complex-linear separation by $W$ is therefore equivalent to separation by $W_{\rm R},W_{\rm I}$.  If $X\ne0$ is Hermitian and traceless in the annihilator and $\rho_0$ is full rank, then for sufficiently small $0<|\epsilon|$ the operators $\rho_0\pm\epsilon X$ are distinct density matrices satisfying
\[
 \Tr\!\left((\rho_0+\epsilon X)O\right)
 =
 \Tr\!\left((\rho_0-\epsilon X)O\right)
 \qquad
 \text{for every accessible }O.
\]

No modes beyond $|n|=N$ on $H_N$ or $|n|=L$ on $H_{\le L}$ add independent compressed directions \cref{thm:finite-generator,cor:finite-bandwidth,cor:window-finite-bandwidth}; powers of the nonscalar $L_0$ remain in the window.  Narrower bands define different inverse problems.  We assume exact asymptotic correlators and omit resolution, noise, contour implementation, and detector dynamics.

\section{Finite moment operator systems}
\label{sec:operator-systems}

\subsection{PBW and Shapovalov conventions}

The Virasoro generators obey
\begin{equation}
 [L_m,L_n]=(m-n)L_{m+n}+\frac c{12}m(m^2-1)\delta_{m+n,0}.
\label{eq:virasoro}
\end{equation}
Choose $\ket h$ with $L_0\ket h=h\ket h$ and $L_n\ket h=0$ for $n>0$.

The canonical contravariant Shapovalov form $\langle\cdot|\cdot\rangle_{\rm Sh}$ is uniquely defined by the normalization condition $\langle h|h\rangle_{\rm Sh}=1$ and the contravariance property $\langle L_n v,w\rangle_{\rm Sh}=\langle v,L_{-n}w\rangle_{\rm Sh}$. At arbitrary algebraic $(c,h)$, the Shapovalov form need not even be a positive-definite physical inner product.

A convenient way to index representatives of classes of level-$N$ descendants is to label them by partitions $\lambda=(\lambda_1\ge\lambda_2\ge\cdots\ge\lambda_\ell)$ where $\lambda_1+\cdots+\lambda_\ell=N$. We write $\lambda\vdash N$ (also $\lambda\in{\mathbb P}_N$) to indicate that $\lambda$ is a partition of the integer $N$. The integer $N$ is called the *weight* of $\lambda$, i.e., $|\lambda|=N$, and the number of *parts* of $\lambda$ is called its *length*, $\ell(\lambda)=\ell$. The unique partition of weight 0 is the *empty partition*, denoted $\varnothing$; its length is also zero. For each $N$, the number of integer partitions $\lambda$ of $N$, i.e., the number of different ways $N$ can be written as the sum of positive integers (where the order doesn't matter), is denoted $p(N)$. Set
\begin{equation}
 \ket\lambda_{\PBW}=L_{-\lambda_1}\cdots L_{-\lambda_\ell}\ket h.
\end{equation}
These $p(N)$ vectors form the Poincar\'e--Birkhoff--Witt (PBW) basis of $H_N=V_{c,h}[N]$.

Two adjunctions occur.  Algebraically, $\dagger$ is the linear Shapovalov anti-involution fixing scalars, reversing products, and sending $L_n$ to $L_{-n}$.  At real unitary parameters it becomes the Hilbert adjoint after scalar conjugation.  In the PBW basis, the Gram matrix is
\begin{equation}
G_{\lambda\mu}
={}_{\PBW}\!\langle\lambda|\mu\rangle_{\PBW}.
\end{equation}
For invertible $G$, coefficient-dual covectors are represented by
\begin{equation}
\bra{\lambda^\vee}
 =\sum_{\nu\vdash N}(G^{-1})_{\lambda\nu}\,{}_{\PBW}\!\bra\nu
\label{eq:dual-from-Gram}
\end{equation}
For nonsingular $G$, this gives
\begin{equation}
 \braket{\lambda^\vee}{\mu}_{\PBW}=\delta_{\lambda\mu}.
\end{equation}
Rank is coordinate-independent, but PBW coordinates expose the partition combinatorics used below.

\subsection{Rank and saturation}

\begin{definition}[Finite stress-tensor operator system]
The chiral measurement space on $H_N$ accessible through at most $k$ stress-tensor insertions is
\begin{equation}
 A_k(N)=P_N\Span\left\{L_{n_1}\cdots L_{n_r}:0\le r\le k,\ \sum_{i=1}^{r}n_i=0\right\}P_N\subseteq\End(H_N).
\label{eq:Ak-def}
\end{equation}
The span is Shapovalov-adjoint stable at real parameters and Hilbert-adjoint stable at unitary points.  Thus ``operator system'' means a unital, adjoint-stable linear space, not necessarily an algebra.
\end{definition}

\begin{definition}[Rank and saturation]
Set
\[
 r_k(N)=\dim A_k(N).
\]
The order-$k$ record on $H_N$ is complete iff $A_k(N)=\End(H_N)$; $k_{\rm sat}(N)$ is the least such order.
\end{definition}
We formalize the window system and trace annihilator from \cref{sec:intro}.
\begin{definition}[Finite-window stress-tensor operator system]
\label{def:finite-window-system}
For $L\ge0$, let
\[
 H_{\le L}=\bigoplus_{n=0}^{L}H_n,
 \qquad
 Q_L=\sum_{n=0}^{L}P_n.
\]
The chiral measurement space on this window through order $k$ is
\begin{equation}
 A_k^{\le L}
 =Q_L\Span\left\{L_{n_1}\cdots L_{n_r}:0\le r\le k,\ n_i\in\mathbb Z\right\}Q_L
 \subseteq\End(H_{\le L}).
\label{eq:A-window-def}
\end{equation}
No zero-total-mode constraint applies.  Let $k_{\rm sat}^{\le L}$ be the least $k$ with $A_k^{\le L}=\End(H_{\le L})$.  Compression by $Q_L$ encodes the support promise; level projectors are not observables.
\end{definition}

\begin{definition}[Trace annihilator]\label{def:invisible}
Given a finite-dimensional space $\mathcal H$ and a linear subspace $U\subseteq\End(\mathcal H)$, define the trace annihilator by
\begin{equation}
 %\Ann_{\Tr}(U)
I_k(N)(U)
 =\{X\in\End(\mathcal H):\Tr(XO)=0\text{ for every }O\in U\}.
\end{equation}
\end{definition}

\section{Finite normal-order theorem}
\label{sec:finite-generator}

\begin{lemma}[Normal ordering respects the filtration degree]
\label{lem:normal-order-length}
Every word $L_{n_1}\cdots L_{n_r}$ is a finite linear combination of normal monomials
\begin{equation}
 L_{-\lambda_1}\cdots L_{-\lambda_a}L_0^bL_{\mu_1}\cdots L_{\mu_d}
\end{equation}
with $a+b+d\le r$, and each term has the same total mode number as the original word.
\end{lemma}

\begin{proof}
For an adjacent inversion use $L_mL_n=L_nL_m+[L_m,L_n]$.  Reordering preserves length and lowers inversion count; the possibly central commutator is shorter with the same total mode number.  Lexicographic induction terminates without longer terms.
\end{proof}

\begin{theorem}[Finite normal-order generator theorem]
\label{thm:finite-generator}
For every $N,k\ge0$,
\begin{equation}
 A_k(N)=\Span\left\{P_NL_{-\lambda_1}\cdots L_{-\lambda_a}L_{\mu_1}\cdots L_{\mu_b}P_N:
 |\lambda|=|\mu|\le N,\ \ell(\lambda)+\ell(\mu)\le k\right\},
\label{eq:finite-generator}
\end{equation}
where the empty pair of partitions denotes the identity.
\end{theorem}

\begin{proof}
Normal-ordering a zero-total word of length at most $k$ gives terms
\begin{equation}
 M=L_{-\lambda_1}\cdots L_{-\lambda_a}L_0^bL_{\mu_1}\cdots L_{\mu_d},
 \qquad |\lambda|=|\mu|,
 \qquad a+b+d\le k.
\end{equation}
The positive block first lowers level by $|\mu|$; if $|\mu|>N$, it kills $H_N$.  For $|\mu|\le N$, $L_0^b$ acts by $(h+N-|\mu|)^b$, giving \eqref{eq:finite-generator}.  Conversely, every listed operator comes from a zero-total word of length at most $k$.
\end{proof}

Let $A_{k,B}(N)$ be the span in \eqref{eq:Ak-def} with $|n_i|\le B$, separating insertion order from Fourier bandwidth.

\begin{corollary}[Finite sufficient Fourier bandwidth]
\label{cor:finite-bandwidth}
For every $N,k\ge0$,
\begin{equation}
 A_{k,N}(N)=A_k(N).
\label{eq:finite-bandwidth}
\end{equation}
If $B\ge N$, then $A_{k,B}(N)=A_k(N)$; no harmonic beyond $|n|=N$ carries independent information on $H_N$.
\end{corollary}

\begin{proof}
By \eqref{eq:finite-generator}, projection bounds partition weights by $N$, hence $|n|\le N$ and $A_k(N)\subseteq A_{k,N}(N)$.  The reverse inclusion is definitional.
\end{proof}

\begin{corollary}[Finite-window normal form and bandwidth]
\label{cor:window-finite-bandwidth}
For every $L,k\ge0$,
\begin{equation}
 A_k^{\le L}
 =
 \Span\left\{
 Q_LL_{-\lambda_1}\cdots L_{-\lambda_a}L_0^b
 L_{\mu_1}\cdots L_{\mu_d}Q_L:
 \begin{array}{c}
 |\lambda|\le L,\ |\mu|\le L,\\
 a+b+d\le k
 \end{array}
 \right\},
\label{eq:window-finite-generator}
\end{equation}
Here $\lambda$ and $\mu$ include the empty partition; modes with $|n|>L$ carry no independent information on $H_{\le L}$.
\end{corollary}

\begin{proof}
By \cref{lem:normal-order-length}, a nonzero block from $H_s$ to $H_r$, with $0\le r,s\le L$, requires $|\mu|\le s$ and
\[
 r=s-|\mu|+|\lambda|,
\]
so
\[
 |\lambda|=r-s+|\mu|\le r\le L,
 \qquad
 |\mu|\le s\le L.
\]
Thus surviving modes have magnitude at most $L$.  Since $L_0$ is not scalar on the window, its powers remain; the definition of $A_k^{\le L}$ gives the reverse inclusion.
\end{proof}

\begin{corollary}
\label{cor:A0A1A2}
For every $N$, $A_0(N)=A_1(N)=\mathbb CI_N$, and
\begin{equation}
 A_2(N)\subseteq\Span\{I_N,P_NL_{-1}L_1P_N,\ldots,P_NL_{-N}L_NP_N\}.
\label{eq:A2-upper}
\end{equation}
Consequently, its dimension is at most $N+1$.
\end{corollary}

\begin{proof}
At length one only $L_0$ preserves level, acting scalarly on $H_N$.  A nontrivial normal length-two word has equal-magnitude negative and positive modes, hence $\lambda=\mu=(d)$; projection requires $1\le d\le N$.
\end{proof}

\section{Bilinear stress-tensor moments}
\label{sec:length-two}

On $H_N$, scalar $L_0$ leaves zeroth and first moments with only normalization and fixed energy; bilinears are first state-dependent.  By \cref{cor:A0A1A2}, they span at most $N+1$ directions on a space of dimension $p(N)$; Appendix~\ref{app:length-two-proof} proves that this upper bound is attained generically.

\begin{theorem}[Generic length-two rank theorem]
\label{thm:length-two}
For every fixed $N\ge2$, some nonempty Zariski-open set $U_N\subset\mathbb C^{2}$ has the property that $(c,h)\in U_N$ implies
\begin{equation}
 \dim A_2(N)=N+1.
\label{eq:length-two-rank}
\end{equation}
For $N=1$, instead,
\begin{equation}
 \dim A_2(1)=1.
\end{equation}
\end{theorem}

\begin{proof}[Outline]
The upper bound is \eqref{eq:A2-upper}.  Appendix~\ref{app:length-two-proof} exhibits a nonzero Feigin--Fuchs/Heisenberg minor for the identity and $P_NL_{-n}L_nP_N$, $1\le n\le N$.  Dependence is Zariski closed in $(c,h)$, proving generic independence.  At $N=1$, $p(1)=1$, so $A_2(1)$ and $\End(H_1)$ are one-dimensional.
\end{proof}

After normalization, bilinears give at most $N$ real numbers on $H_N$, against $p(N)^2-1$ parameters for a trace-one Hermitian state.  Individual $L_{-n}L_n$ detect selected PBW coherences, but the span retains a large trace annihilator.

\section{Level-two phase witness}
\label{sec:N2-example}

Throughout the continuous unitary regime, the following level-two witness has an undetected relative-phase sign and a cubic separator.

\begin{proposition}[Generic level-two phase witness]
\label{prop:N2-generic-heavy}
Let $c>1$ and $h>0$, and set
\begin{equation}
 D=4h(2h+1),
 \qquad
 Q=4h+\frac c2-\frac{9h}{2h+1}.
\label{eq:N2-generic-DQ}
\end{equation}
Then $D$ and $Q$ are positive, and the vectors
\begin{equation}
 \ket{e_d}=\frac{L_{-1}^2\ket h}{\sqrt D},
 \qquad
 \ket{e_q}=\frac{L_{-2}\ket h-\dfrac{3}{2(2h+1)}L_{-1}^2\ket h}{\sqrt Q}
\label{eq:N2-generic-basis}
\end{equation}
form an orthonormal pair.  Consider the phase-coherent superpositions and the corresponding incoherent mixture
\begin{equation}
 \ket{\Psi_\pm}=\frac{\ket{e_q}\pm i\ket{e_d}}{\sqrt2},
 \qquad
 \rho_{\rm mix}=\frac12\left(\ket{e_q}\bra{e_q}+\ket{e_d}\bra{e_d}\right)
\label{eq:N2-generic-states}
\end{equation}
Every $O\in A_2(2)$ has the same expectation value in $\ket{\Psi_+}$, $\ket{\Psi_-}$, and $\rho_{\rm mix}$.  At cubic order the phase is resolved by
\begin{equation}
 H=\frac{L_{-2}L_1^2-L_{-1}^2L_2}{2i},
 \qquad
 \bra{\Psi_\pm}H\ket{\Psi_\pm}=\pm\frac{\sqrt{DQ}}2,
 \qquad
 \Tr(\rho_{\rm mix}H)=0.
\label{eq:N2-generic-witness}
\end{equation}
At heavy chiral weight the effect persists: if $c\gg1$ and $h-c/24=O(c)>0$, the same direction and separator remain.  In a semiclassical Einstein holographic family, an antiholomorphic sector with the analogous heavy-energy condition places the pair in the macroscopic nonextremal BTZ regime.
\end{proposition}

\begin{proof}
In the ordered PBW basis $(L_{-2}\ket h,L_{-1}^2\ket h)$, the level-two Gram matrix is
\begin{equation}
 G_2(c,h)=
 \begin{pmatrix}
  4h+c/2&6h\\
  6h&4h(2h+1)
 \end{pmatrix}.
\end{equation}
For $c>1$ and $h>0$, continuous-series unitarity makes the matrix positive definite \cite{KacRaina:1987,Friedan:1984,DiFrancesco:1997}.  Gram--Schmidt gives \eqref{eq:N2-generic-basis}; its numerators have squared norms $D$ and $Q$, and $\det G_2=DQ$.

Writing $q=\sqrt Q\ket{e_q}$ and $d=\sqrt D\ket{e_d}$, the Virasoro commutators give, before normalization,
\begin{equation}
 L_1q=0,
 \qquad
 L_1d=2(2h+1)L_{-1}\ket h,
 \qquad
 L_2q=Q\ket h,
 \qquad
 L_2d=6h\ket h,
\end{equation}
so that, in the ordered basis $(e_q,e_d)$, the two nontrivial bilinears are
\begin{equation}
 L_{-1}L_1=
 \begin{pmatrix}0&0\\0&2(2h+1)\end{pmatrix},
 \qquad
 L_{-2}L_2=
 \begin{pmatrix}
  Q&6h\sqrt{Q/D}\\
  6h\sqrt{Q/D}&\dfrac{9h}{2h+1}
 \end{pmatrix}.
\label{eq:N2-generic-bilinears}
\end{equation}
With the identity, these matrices span the complex-symmetric $2\times2$ matrices and, by \cref{cor:A0A1A2}, exhaust $A_2(2)$.  They miss the sign of the imaginary coherence in \eqref{eq:N2-generic-states}.  The antisymmetric direction appears next: $L_1^2q=0$, whereas $L_1^2d=D\ket h$, and hence
\begin{equation}
 L_{-2}L_1^2=
 \begin{pmatrix}0&\sqrt{DQ}\\0&6h\end{pmatrix},
 \qquad
 H=\frac{\sqrt{DQ}}2\,\sigma_y.
\end{equation}
The states $\ket{\Psi_\pm}$ are the eigenvectors of $\sigma_y$, with the opposite values in \eqref{eq:N2-generic-witness}.
\end{proof}

\section{Exact saturation order}
\label{sec:saturation}

Write $G_j(c,h)$ for the Shapovalov Gram matrix at level $j$, and introduce the regular locus
\begin{equation}
 \cU_N^{\rm reg}=\left\{(c,h)\in\mathbb C^2:\det G_j(c,h)\ne0\text{ for }0\le j\le N\right\}.
\label{eq:regular-locus}
\end{equation}
This locus is nonempty and Zariski open.  We prove saturation over $\mathbb Q(c,h)$ and then specialize to $\cU_N^{\rm reg}$.

\begin{proposition}[Heavy unitary regular locus]
\label{prop:heavy-regular}
Fix $N$ and a real $c>1$.  The simultaneous regularity conditions
\begin{equation}
 \det G_j(c,h)\ne0,
 \qquad 1\le j\le N,
\end{equation}
fail at only finitely many values of $h$.  Thus $\cU_N^{\rm reg}$ intersects the continuous unitary highest-weight regime at arbitrarily large positive conformal weight.
\end{proposition}

\begin{proof}
At level $j$, the Kac determinant is \cite[Ch.~8]{KacRaina:1987}
\begin{equation}
 \det G_j(c,h)=C_j
 \prod_{\substack{r,s\ge1\\rs\le j}}
 \bigl(h-h_{r,s}(c)\bigr)^{p(j-rs)}.
\label{eq:kac-product}
\end{equation}
For fixed $c$ this nonzero polynomial in $h$ has finitely many roots.  Their union over $1\le j\le N$ is finite, with unbounded real complement.  The continuous unitary family includes $c\ge1$, $h\ge0$ \cite{KacRaina:1987,Friedan:1984,DiFrancesco:1997}; choosing large positive $h$ off that set gives a positive-definite module regular through the required levels.
\end{proof}

The theorem uses no expansion in $c$ or $h$.  A macroscopic nonextremal Einstein--BTZ reading additionally requires $c=\bar c=3\ell/(2G)\gg1$, positive cylinder energies $h-c/24$ and $\bar h-\bar c/24$ of order $c$, and a heavy holographic family with a justified semiclassical bulk description.  The hierarchies $N,L=o(c)$ and $\bar N,\bar L=o(\bar c)$ make descendant energy subleading; they neither enter the algebraic result nor suffice for perturbativity, which also depends on occupied modes and stress profiles, while regularity and unitarity remain separate chiral assumptions.  On $H_N\otimes\bar H_{\bar N}$, zero-mode eigenvalues $h+N$ and $\bar h+\bar N$ also fix the ADM charges in \eqref{eq:btz-descendant-charges}.  The weights $(h,\bar h)$ remain input data.

The generation argument turns on the vector Shapovalov-dual to the all-ones descendant.  With $\delta_j=(1^j)$, define $D_j^\#\in H_j$ by
\begin{equation}
 G_j(D_j^\#,\ket\lambda_{\PBW})=\delta_{\lambda,\delta_j}.
\label{eq:Dsharp-main}
\end{equation}
Contravariance yields the Whittaker-chain relations
\begin{equation}
 L_1D_j^\#=D_{j-1}^\#,
 \qquad
 L_mD_j^\#=0\quad(m\ge2).
\label{eq:Whittaker-main}
\end{equation}
\begin{theorem}[Exact saturation order]
\label{thm:saturation}
For fixed $N\ge2$ and every $(c,h)\in\cU_N^{\rm reg}$,
\begin{align}
\dim I_k(N)&\ge\left\lfloor\frac N2\right\rfloor,
 \label{eq:saturation-ann-bound}\\
 A_N(N)&\subsetneq\End(H_N),
 \label{eq:saturation-lower}\\
 A_{N+1}(N)&=\End(H_N).
 \label{eq:saturation-upper}
\end{align}
The exact threshold is
\begin{equation}
 k_{\rm sat}(N)=N+1.
\label{eq:ksat-theorem}
\end{equation}
For $N=0,1$, $H_N$ is one-dimensional and $k_{\rm sat}(N)=0$.
\end{theorem}

\begin{proof}[Proof Outline]
Details are in Appendix~\ref{app:saturation-proof}.

For the order-$N$ obstruction, lower one index of an alternating matrix with the Shapovalov form and pair the tensor $B$ with $O\in\End(H_N)$ by $\Tr(G_N^{-1}BO)$.  The star at $(1^N)$, with hook $(N-1,1)$ omitted, has $p(N)-2$ variables for $N\ge3$.  By \eqref{eq:Whittaker-main}, every order-$N$ normal-generator constraint lies in the row span of
\begin{equation}
 r_\alpha^{(N)}(\lambda)=\bra{\lambda^\vee}L_{-\alpha}D_{N-|\alpha|}^\#,
 \qquad
 2\le|\alpha|\le N-2,
 \quad
 \alpha\ne(1^{|\alpha|}),
 \quad
 |\alpha|+\ell(\alpha)\le N.
\end{equation}
There are exactly $p(N)-2-\lfloor N/2\rfloor$ rows, so rank--nullity leaves at least $\lfloor N/2\rfloor$ independent star-supported functionals in the trace annihilator of $A_N(N)$.  Thus order $N$ is not tomographically complete.

For the upper bound, the vectors
\begin{equation}
 L_{-\alpha}D_{N-|\alpha|}^\#,
 \qquad
 |\alpha|+\ell(\alpha)\le N,
\end{equation}
form a basis of $H_N$, so $D_N^\#$ is cyclic under $A_N(N)$.  Unequal-level transition density, followed by the diagonal lift, fills all maps into the codimension-one PBW hyperplane spanned by partitions other than $(1^N)$; cyclicity supplies the quotient maps.  Hence $A_{N+1}(N)=\End(H_N)$.
\end{proof}

For nonnegative $r,s,k$, let the mode-length filtered transition space be
\begin{equation}
 \cT_k(r,s)=\Span\left\{P_rL_{n_1}\cdots L_{n_t}P_s:
 0\le t\le k,\ \sum_i n_i=s-r\right\}
 \subseteq\operatorname{Hom}(H_s,H_r).
\label{eq:transition-space}
\end{equation}
Its diagonal restriction is $\cT_k(N,N)=A_k(N)$.

\begin{lemma}[Transition density]
\label{lem:transition-density}
For $r\ne s$ and every $(c,h)\in\cU_{\max(r,s)}^{\rm reg}$,
\begin{equation}
 \cT_{\max(r,s)}(r,s)=\operatorname{Hom}(H_s,H_r).
\end{equation}
For every $(c,h)\in\cU_N^{\rm reg}$,
\begin{equation}
 \cT_{N+1}(N,N)=\End(H_N).
\end{equation}
\end{lemma}

\begin{proof}
See Appendix~\ref{app:saturation-proof}.
\end{proof}

\begin{corollary}[Boundary-observer tomography threshold]
\label{cor:operational-threshold}
Let $N\ge2$, with $(c,h)\in\cU_N^{\rm reg}$ unitary, and define the Hermitian invisible subspace
\begin{equation}
 \mathfrak I_N^{\rm herm}
 =\left\{X\in\End(H_N):X=X^\dagger,\ \Tr(XO)=0\ \text{for every }O\in A_N(N)\right\}.
\end{equation}
Its real dimension satisfies
\begin{equation}
 \dim_{\mathbb R}\mathfrak I_N^{\rm herm}
 \ge \left\lfloor\frac N2\right\rfloor.
\end{equation}
Every $X\in\mathfrak I_N^{\rm herm}$ is traceless.  If $X\ne0$ and
\begin{equation}
 0<\varepsilon<\frac{1}{p(N)\lVert X\rVert_{\rm op}},
\end{equation}
then
\begin{equation}
 \rho_\pm=\frac{I_N}{p(N)}\pm\varepsilon X
\label{eq:indistinguishable-states}
\end{equation}
are distinct density matrices with identical order-$N$ data:
\begin{equation}
 \cM_N(\rho_+)=\cM_N(\rho_-).
\end{equation}
One further stress-tensor insertion removes the ambiguity:
\begin{equation}
 \cM_{N+1}(\rho)=\cM_{N+1}(\sigma)
 \quad\Longrightarrow\quad
 \rho=\sigma.
\end{equation}
\end{corollary}

\begin{proof}
Adjoint stability of $A_N(N)$ passes to its annihilator.  Hermitianization equates its complex dimension with $\dim_{\mathbb R}\mathfrak I_N^{\rm herm}$; \cref{thm:saturation} gives the bound.  Since $I_N\in A_N(N)$, invisible Hermitian directions are traceless.  The bound in \eqref{eq:indistinguishable-states} makes $\rho_\pm$ positive and unit trace, and annihilation gives identical records.  At order $N+1$, \cref{thm:saturation} and trace duality separate them.
\end{proof}

\begin{corollary}[Pure-state hiding at the penultimate order]
\label{cor:pure-state-hiding}
Let $N\ge2$ and suppose $(c,h)\in\cU_N^{\rm reg}$ is unitary.  The loss of injectivity at order $N$ already occurs among pure states.  More strongly, the Hermitian trace annihilator contains a real linear subspace
\begin{equation}
 \mathfrak K_N\subseteq\mathfrak I_N^{\rm herm},
 \qquad
 \dim_{\mathbb R}\mathfrak K_N\ge\left\lfloor\frac N2\right\rfloor,
\end{equation}
every nonzero member of which has rank two, with one positive and one negative eigenvalue on its support.  Given $0\ne X\in\mathfrak K_N$, choose unit eigenvectors $\ket{\Psi_+},\ket{\Psi_-}\in H_N$ for those eigenvalues and set
\begin{equation}
 \sigma_\pm=\ket{\Psi_\pm}\bra{\Psi_\pm},
 \qquad
 \sigma_{\rm mix}=\frac12(\sigma_++\sigma_-).
\end{equation}
All observables formed from at most $N$ stress-tensor insertions assign the same expectation value to these three states:
\begin{equation}
 \cM_N(\sigma_+)=\cM_N(\sigma_-)=\cM_N(\sigma_{\rm mix}).
 \label{eq:pure-state-hiding}
\end{equation}
The orthogonal pure states are separated at order $N+1$: $A_{N+1}(N)$ contains a norm-one Hermitian operator $Z$ satisfying
\begin{equation}
 \bra{\Psi_+}Z\ket{\Psi_+}=1,
 \qquad
 \bra{\Psi_-}Z\ket{\Psi_-}=-1,
 \qquad
 \Tr(\sigma_{\rm mix}Z)=0.
 \label{eq:pure-state-separator}
\end{equation}
\end{corollary}

\begin{proof}
Let $\cB_{N,\mathbb R}^\star=\Span_{\mathbb R}\{B_\lambda:\lambda\in\Lambda_N^\star\}$ be the real star block and $R_{N,\mathbb R}$ the specialized pairing restriction from Appendix~\ref{app:saturation-proof}.  At a unitary point the PBW constants and $G_N$ are real, so $R_{N,\mathbb R}$ has unchanged rank after complexification.  For $N\ge3$, the star-row count gives $\dim_{\mathbb R}\ker R_{N,\mathbb R}\ge\lfloor N/2\rfloor$; the separate $N=2$ calculation gives one real direction.  Every $B\in\ker R_{N,\mathbb R}$ has the form
\begin{equation}
 B=\sum_{\lambda\in\Lambda_N^\star}b_\lambda
 \bigl(E_{\lambda,\delta_N}-E_{\delta_N,\lambda}\bigr)
 =b\,e_{\delta_N}^{\mathsf T}-e_{\delta_N}b^{\mathsf T},
 \label{eq:rank-two-star-form}
\end{equation}
Here $e_{\delta_N}$ is the PBW column for $\delta_N$, while $b$ has zero $\delta_N$ coordinate.  If $B\ne0$, its image lies in $\Span_{\mathbb R}\{b,e_{\delta_N}\}$; both directions occur since $Be_{\delta_N}=b$ and $Bb=-(b^{\mathsf T}b)e_{\delta_N}$, with $b^{\mathsf T}b>0$ for nonzero real $b$.  Thus $B$ has rank two.

Raising the index gives $Y_B=G_N^{-1}B$.  Regularity preserves rank, and the kernel condition is $\Tr(Y_BO)=0$ for every $O\in A_N(N)$.  Hence $Y_B$ is a rank-two member of %$\Ann_{\Tr}(A_N(N))$.
$I_k(N)$.

Only here is unitarity needed.  The real symmetric positive-definite Shapovalov matrix is the Hilbert metric on $H_N$.  Since $B$ and $Y_B$ are real, the alternating tensor is anti-Hermitian:
\begin{equation}
 Y_B^\dagger
 =G_N^{-1}Y_B^{\mathsf T}G_N
 =G_N^{-1}B^{\mathsf T}
 =-Y_B.
\end{equation}
Thus $X_B=iY_B$ is Hermitian, rank two, and invisible to $A_N(N)$.  The real-linear map $B\mapsto iG_N^{-1}B$ is injective on $\ker R_{N,\mathbb R}$ and defines $\mathfrak K_N$.  Since $I_N\in A_N(N)$, $\Tr X_B=0$; its nonzero eigenvalues are $a$ and $-a$ for some $a>0$.

For $B\ne0$, let $\ket{\Psi_\pm}$ be normalized eigenvectors of $X_B$ with eigenvalues $\pm a$.  On the support,
\begin{equation}
 X_B=a(\sigma_+-\sigma_-).
\end{equation}
For every $O\in A_N(N)$, annihilation gives
\begin{equation}
 0=\Tr(X_BO)
 =a\left(\bra{\Psi_+}O\ket{\Psi_+}
 -\bra{\Psi_-}O\ket{\Psi_-}\right),
\end{equation}
Since $a>0$, both pure states, and by linearity their equal mixture, have the same expectation.

At the next order, \cref{thm:saturation} gives $A_{N+1}(N)=\End(H_N)$, including the projector difference
\begin{equation}
 Z=\sigma_+-\sigma_-
\end{equation}
which belongs to $A_{N+1}(N)$, acts on $\ket{\Psi_\pm}$ with eigenvalues $\pm1$, vanishes on $\sigma_{\rm mix}$, and obeys $\lVert Z\rVert_{\rm op}=1$.
\end{proof}

On the support of $X_B$, the eigenvectors can be written as opposite-sign equal-population superpositions; deleting that coherence gives $\sigma_{\rm mix}$ with unchanged order-$N$ moments.  This is a basis description, not physical dephasing.  Invariantly, the three records agree at order-$N$; $Z\in A_{N+1}(N)$ asserts span membership, not a chosen word or synthesis cost.

\section{Finite descendant windows}
\label{sec:finite-window}

A window drops exact descendant energy: states may span levels, with off-diagonal $L_0$ blocks carrying coherence.  We retain only $\rho=Q_L\rho Q_L$.  The system $A_k^{\le L}$ from \cref{def:finite-window-system} therefore includes nonzero-total-mode words.

For each integer $q$, let
\begin{equation}
 A_{k,q}^{\le L}
 =
 Q_L\Span\left\{
 L_{n_1}\cdots L_{n_t}:
 0\le t\le k,\ \sum_{i=1}^{t}n_i=q
 \right\}Q_L
\end{equation}
and the corresponding grade-$q$ block band
\begin{equation}
 E_q^{\le L}
 =
 \bigoplus_{\substack{0\le r,s\le L\\s-r=q}}
 \operatorname{Hom}(H_s,H_r).
\label{eq:window-grade-block}
\end{equation}
If $W$ has total mode number $q$, then
\begin{equation}
 [L_0,W]=-qW,
 \qquad
 WH_s\subseteq H_{s-q}.
\label{eq:window-grade-shift}
\end{equation}
The $L_0$ grading decomposes both spaces into disjoint block bands:
\begin{equation}
 \End(H_{\le L})
 =\bigoplus_{q=-L}^{L}E_q^{\le L},
 \qquad
 A_k^{\le L}
 =\bigoplus_{q=-L}^{L}A_{k,q}^{\le L}.
\label{eq:window-grade-decomposition}
\end{equation}

\begin{theorem}[Exact finite-window saturation]
\label{thm:finite-window-saturation}
Let $L\ge2$ and $(c,h)\in\cU_L^{\rm reg}$.  Then
\begin{equation}
 A_L^{\le L}\subsetneq\End(H_{\le L}),
 \qquad
 A_{L+1}^{\le L}=\End(H_{\le L}).
\label{eq:finite-window-saturation}
\end{equation}
The saturation order is therefore
\begin{equation}
 k_{\rm sat}^{\le L}=L+1.
\end{equation}
The two regular low-level cases are
\begin{equation}
 k_{\rm sat}^{\le0}=0,
 \qquad
 k_{\rm sat}^{\le1}=1.
\end{equation}
\end{theorem}

\begin{proof}
The upper bound follows from the stronger grade-wise identity
\begin{equation}
 A_{L+1,q}^{\le L}=E_q^{\le L},
 \qquad
 -L\le q\le L,
\label{eq:window-grade-fullness}
\end{equation}
which is proved by induction on $L$.

At $L=0$ the space is a line, so the empty word suffices.  At $L=1$, regularity means $\det G_1(c,h)=2h\ne0$.  The distinct $L_0$ eigenvalues $h$ and $h+1$ on $H_0$ and $H_1$ make $Q_1IQ_1$ and $Q_1L_0Q_1$ span the diagonal units; $Q_1L_{-1}Q_1$ and $Q_1L_1Q_1$ give the off-diagonal units, which are nonzero because
\begin{equation}
 L_1L_{-1}\ket h=2h\ket h.
\end{equation}
Thus $A_1^{\le1}=\End(H_{\le1})$.

Take $L\ge2$ and assume
\begin{equation}
 A_L^{\le L-1}=\End(H_{\le L-1}).
\label{eq:window-induction-hypothesis}
\end{equation}
For $L=2$ this follows from the base case and monotonicity; for $L\ge3$ it is induction at cutoff $L-1$.

Fix $|q|\le L$ and $X\in E_q^{\le L}$.  Exactly one source--target pair on the grade-$q$ block diagonal meets the top level; denote it by $(r_{\rm t},s_{\rm t})$, where
\begin{equation}
 (r_{\rm t},s_{\rm t})
 =
 \begin{cases}
 (L-q,L),&q\ge0,\\[1mm]
 (L,L+q),&q<0.
 \end{cases}
\label{eq:window-top-pair}
\end{equation}
In both cases,
\begin{equation}
 s_{\rm t}-r_{\rm t}=q,
 \qquad
 \max(r_{\rm t},s_{\rm t})=L.
\end{equation}

For $q\ne0$, the two levels differ, and \cref{lem:transition-density} gives
\begin{equation}
 \cT_L(r_{\rm t},s_{\rm t})
 =\operatorname{Hom}(H_{s_{\rm t}},H_{r_{\rm t}}).
\end{equation}
Thus a homogeneous combination $W_{\rm t}$ of grade-$q$ words of length at most $L$ satisfies
\begin{equation}
 P_{r_{\rm t}}W_{\rm t}P_{s_{\rm t}}
 =P_{r_{\rm t}}XP_{s_{\rm t}}.
\label{eq:window-top-match}
\end{equation}
For $q=0$, $(r_{\rm t},s_{\rm t})=(L,L)$, and \cref{thm:saturation} supplies a grade-zero $W_{\rm t}$ of length at most $L+1$ satisfying \eqref{eq:window-top-match}.  Thus every top block is matched within length $L+1$.

Every remaining grade-$q$ pair lies in $H_{\le L-1}$ and has source level $s\ne s_{\rm t}$.  On these lower blocks define $Y\in E_q^{\le L-1}$ by
\begin{equation}
 P_rYP_s
 =
 \frac{P_rXP_s-P_rW_{\rm t}P_s}{s-s_{\rm t}},
 \qquad
 \substack{0\le r,s\le L-1\\s-r=q}.
\label{eq:window-lower-residual}
\end{equation}
The denominators in \eqref{eq:window-lower-residual} are nonzero.  By \eqref{eq:window-induction-hypothesis}, a combination $U$ of length at most $L$ satisfies
\begin{equation}
 Q_{L-1}UQ_{L-1}=Y.
\end{equation}
Take $U$ homogeneous of grade $q$: after compression to $H_{\le L-1}$, total-mode components occupy distinct bands by \eqref{eq:window-grade-decomposition} at cutoff $L-1$, so only grade $q$ represents $Y$.

Now set
\begin{equation}
 C=U\bigl(L_0-(h+s_{\rm t})\mathbf 1\bigr).
\label{eq:window-correction}
\end{equation}
Here $\mathbf 1$ is the enveloping-algebra identity.  Then $C$ has grade $q$ and length at most $L+1$.  On each grade-$q$ pair $(r,s)$,
\begin{align}
 P_rCP_s
 &=
 P_rU\bigl(L_0-(h+s_{\rm t})\mathbf 1\bigr)P_s \notag\\
 &=(s-s_{\rm t})P_rUP_s.
\label{eq:window-source-separation}
\end{align}
On the top pair this gives
\begin{equation}
 P_{r_{\rm t}}CP_{s_{\rm t}}=0,
\end{equation}
while on each lower block it gives
\begin{equation}
 P_rCP_s
 =(s-s_{\rm t})P_rYP_s
 =P_rXP_s-P_rW_{\rm t}P_s.
\end{equation}
These blocks exhaust the grade-$q$ band, hence
\begin{equation}
 Q_L(W_{\rm t}+C)Q_L=X.
\end{equation}
The representative has grade $q$ and length at most $L+1$.  Thus $E_q^{\le L}\subseteq A_{L+1,q}^{\le L}$; \eqref{eq:window-grade-shift} gives the reverse inclusion.  Summing over $q$ proves the upper equality in \eqref{eq:finite-window-saturation}.

For sharpness, compress the order-$L$ window system to its top level:
\begin{equation}
 P_LA_L^{\le L}P_L=A_L(L).
\label{eq:window-top-corner}
\end{equation}
Only zero-total-mode words contribute to $H_L\to H_L$, so \eqref{eq:window-top-corner} is the fixed-level system.  If $A_L^{\le L}$ equaled $\End(H_{\le L})$, its corner would be $\End(H_L)$, contradicting $A_L(L)\subsetneq\End(H_L)$ in \cref{thm:saturation}.  Thus $A_L^{\le L}$ is proper.
\end{proof}

The projectors identify blocks but do not enter the measurement word, which is the compression of the unprojected combination $W_{\rm t}+C$.  The zero-mode in \eqref{eq:window-correction} separates the top block from lower ones; no spectral projector enters the observable algebra.

At a real unitary specialization, $(A_{k,q}^{\le L})^\dagger=A_{k,-q}^{\le L}$, so $A_k^{\le L}$ is adjoint-stable.  By \cref{thm:finite-window-saturation}, complex-linear fullness is equivalent to completeness with Hermitian observables.  For known support in $H_{\le L}$, order $L+1$ determines every population and intra- or interlevel coherence.

\begin{corollary}[Interlevel coherence and $L_0$-block diagonality]
\label{cor:window-phase}
Let $(c,h)\in\cU_L^{\rm reg}$ be unitary.  If $\ket u\in H_r$ and $\ket v\in H_s$ are unit vectors with $r\ne s$ and $r,s\le L$, then the coherent state
\begin{equation}
 \ket{\psi_\theta}
 =\frac{\ket u+e^{i\theta}\ket v}{\sqrt2}
\end{equation}
can be distinguished from the incoherent mixture
\begin{equation}
 \rho_{\rm mix}
 =\frac12\ket u\bra u+\frac12\ket v\bra v
\end{equation}
by a Hermitian linear combination of allowed stress words in $A_{L+1}^{\le L}$.  For an arbitrary density operator supported in $H_{\le L}$, moreover, all nonzero-grade expectation values through order $L+1$ vanish if and only if
\begin{equation}
 [\rho,L_0]=0.
\end{equation}
\end{corollary}

\begin{proof}
The Hermitian operator
\begin{equation}
 Z_\theta=e^{-i\theta}\ket u\bra v+e^{i\theta}\ket v\bra u
\end{equation}
has value one in $\ket{\psi_\theta}$ and zero in $\rho_{\rm mix}$.  By \cref{thm:finite-window-saturation}, it lies in $A_{L+1}^{\le L}$ and has a Hermitian allowed-word representation.

If $[\rho,L_0]=0$, block diagonality kills every nonzero-grade pairing.  Conversely, \eqref{eq:window-grade-fullness} equates the saturated grade-$q$ system with $E_q^{\le L}$ for each $q\ne0$.  Vanishing on these bands forces every $P_r\rho P_s$, $r\ne s$, to vanish, equivalently $[\rho,L_0]=0$.
\end{proof}

\begin{corollary}[Finite windows and absence of a uniform tower threshold]
\label{cor:finite-support-tower}
Suppose the Shapovalov forms of $V_{c,h}$ are nondegenerate at every level.  For each $L$, order $L+1$ separates all operators on $H_{\le L}$, but no finite order works uniformly over all finite windows.  At a real unitary specialization, this applies to every density operator with known support $\rho=Q_L\rho Q_L$.
\end{corollary}

\begin{proof}
The fixed-window claim is \cref{thm:finite-window-saturation}.  For nonuniformity, fix $k$ and choose $N\ge\max\{2,k\}$.  Compressing order-$k$ words to $H_N$ gives
\begin{equation}
 A_k(N)\subseteq A_N(N)\subsetneq\End(H_N),
\end{equation}
so order $k$ fails already on $H_N$.  At a real unitary point, \cref{cor:pure-state-hiding} gives distinct finite-support states with the same order-$k$ record.
\end{proof}

\section{Nonchiral sharp levels and rectangular windows}
\label{sec:nonchiral}

Restore the antiholomorphic Virasoro algebra and fix both descendant levels.  The symbols $\bar H_{\bar N}$, $\bar A_{\bar k}(\bar N)$, $\bar r_{\bar k}(\bar N)$, and $\bar{\cU}_{\bar N}^{\rm reg}$ denote the antiholomorphic counterparts of the corresponding unbarred objects.  The state space is
\begin{equation}
 H_{N,\bar N}=H_N\otimes\bar H_{\bar N},
\end{equation}
and the commuting actions factorize the separately counted measurement space:
\begin{equation}
 A^{\rm prod}_{k,\bar k}(N,\bar N)=A_k(N)\otimes\bar A_{\bar k}(\bar N),
\end{equation}
and its dimension:
\begin{equation}
 \dim A^{\rm prod}_{k,\bar k}(N,\bar N)=r_k(N)\,\bar r_{\bar k}(\bar N).
\label{eq:prod-rank}
\end{equation}
For $N,\bar N\ge2$, with both parameter pairs in the loci of \cref{thm:length-two}, the bilinear product space has
\begin{equation}
 \dim A^{\rm prod}_{2,2}(N,\bar N)=(N+1)(\bar N+1),
\end{equation}
whereas the full algebra has dimension $p(N)^2p(\bar N)^2$.  As $N$ and $\bar N$ grow, the bilinear left--right record is sparse.

A common insertion budget instead gives
\begin{equation}
 A^{\rm tot}_K(N,\bar N)=\Span\{A_r(N)\otimes\bar A_s(\bar N):r,s\ge0,\ r+s\le K\}.
\end{equation}
Set $\Delta r_j=r_j-r_{j-1}$ and $\Delta\bar r_j=\bar r_j-\bar r_{j-1}$, with $r_{-1}=\bar r_{-1}=0$, and choose successive complements of dimensions $\Delta r_r(N)$ and $\Delta\bar r_s(\bar N)$.  Their degree-$r$ and degree-$s$ product enters exactly when $r+s\le K$; hence
\begin{equation}
 \dim A^{\rm tot}_K(N,\bar N)
 =\sum_{\substack{r,s\ge0\\r+s\le K}}\Delta r_r(N)\,\Delta\bar r_s(\bar N).
\label{eq:total-rank}
\end{equation}

\begin{corollary}[Nonchiral saturation]
\label{cor:nonchiral-saturation}
Suppose $N,\bar N\ge2$, $(c,h)\in\cU_N^{\rm reg}$, and $(\bar c,\bar h)\in\bar{\cU}_{\bar N}^{\rm reg}$.  With separate budgets, the product system is full exactly when $k\ge N+1$ and $\bar k\ge\bar N+1$.  If both chiralities draw on one budget, the saturation threshold is
\begin{equation}
 K^{\rm tot}_{\rm sat}(N,\bar N)=N+\bar N+2.
\end{equation}
\end{corollary}

\begin{proof}
Apply \cref{thm:saturation} to each factor for separate budgets.  With one budget the terminal product has degree $N+\bar N+2$ and is absent for smaller $K$.  At $K=N+\bar N+2$, every increment pair is present and spans $\End(H_N\otimes\bar H_{\bar N})$.
\end{proof}

\begin{corollary}[Nonchiral finite-window saturation]
\label{cor:nonchiral-window-saturation}
Suppose $L,\bar L\ge2$, $(c,h)\in\cU_L^{\rm reg}$, and $(\bar c,\bar h)\in\bar{\cU}_{\bar L}^{\rm reg}$.  On the rectangular support space
\begin{equation}
 H_{\le L,\le\bar L}=H_{\le L}\otimes\bar H_{\le\bar L},
\end{equation}
define
\begin{align}
 A_{k,\bar k}^{\rm prod,win}
 &=A_k^{\le L}\otimes\bar A_{\bar k}^{\le\bar L},\\
 A_K^{\rm tot,win}
 &=\Span\left\{
 A_r^{\le L}\otimes\bar A_s^{\le\bar L}:r,s\ge0,\ r+s\le K
 \right\}.
\end{align}
The product system is full if and only if $k\ge L+1$ and $\bar k\ge\bar L+1$.  With a shared budget, the exact threshold is
\begin{equation}
 K_{\rm sat}^{\rm win}=L+\bar L+2.
\end{equation}
\end{corollary}

\begin{proof}
Apply \cref{thm:finite-window-saturation} to each factor.  With one budget the terminal product has degree $L+\bar L+2$ and is absent for smaller $K$.  At the threshold every increment pair is present.
\end{proof}

The chiral cutoffs are independent; nonzero bidegree records coherence between descendants differing in boundary energy or angular momentum.  In standard cylinder conventions, bidegree $(q,\bar q)$ evolves as
\begin{equation}
 W_{q,\bar q}(t,\phi)
 =e^{-i(q+\bar q)t/\ell}e^{-i(q-\bar q)\phi}W_{q,\bar q}(0,0).
\end{equation}
The phase record requires common boundary time and angular origins; highest-weight sectors are not inferred, so both module labels remain prior data.

\section{Discussion}
\label{sec:discussion}

Order $L$ is incomplete on $H_{\le L}$, whereas $L+1$ spans $\End(H_{\le L})$.  The $H_L$ corner supplies the lower bound.  Other blocks are available by order $L+1$, so the order-$(L+1)$ record fixes the state on $H_{\le L}$ once the module is known.

The block $P_j\rho P_j$ is intralevel data.  At order $L+1$, nonzero-grade moments determine $P_r\rho P_s$ for $r\ne s$; these blocks vanish precisely when $[\rho,L_0]=0$.

Fourier bands below the bounds in \cref{cor:finite-bandwidth,cor:window-finite-bandwidth} define another inverse problem.  Full bandwidth gives exact state determination in one finite window, not efficient reconstruction or a unique semiclassical geometry.

\acknowledgments

This work was supported in part by the Natural Sciences and Engineering Research Council of Canada.

\appendix

\section{Details of the Heisenberg--Vandermonde length-two proof}
\label{app:length-two-proof}

The generic rank statement in \cref{thm:length-two} will follow from a single maximal minor which is not identically zero.  The Feigin--Fuchs realization is used only to exhibit that minor; conventions for the Shapovalov form and the Kac determinant are those of \cite{Feigin:1984,Kac:1979}.  Let $[a_m,a_n]=m\delta_{m+n,0}$, and let $\ket\alpha$ be the Fock vacuum satisfying $a_n\ket\alpha=0$ for $n>0$ and $a_0\ket\alpha=\alpha\ket\alpha$.  With background charge $Q$, the Virasoro parameters are
\begin{equation}
 L_n=\frac12\sum_{k\in\mathbb Z}:a_ka_{n-k}:-Q(n+1)a_n,
 \qquad c=1-12Q^2,
 \qquad h=\frac12\alpha(\alpha-2Q).
\end{equation}
Write the level-$N$ oscillator monomials as
\begin{equation}
\ket{\lambda}_{\FF}=a_{-\lambda_1}\cdots a_{-\lambda_\ell}\ket\alpha.
\end{equation}
On the displayed, unnormalized basis, $\bra{\lambda_{\FF}^{\vee}}$ extracts the coefficient indexed by $\lambda$ in any expansion.  The oscillator Hilbert bra carries multiplicity factors when creation operators are repeated.

The universal map $V_{c,h}\to\mathcal F_\alpha$ sends $\ket h$ to $\ket\alpha$.  At fixed $Q$, the highest power of $\alpha$ in $L_{-n}$ comes from $\alpha a_{-n}$, and consequently
\begin{equation}
 L_{-\lambda_1}\cdots L_{-\lambda_\ell}\ket\alpha
 =\alpha^{\ell(\lambda)}\ket\lambda_{\FF}
 +\text{terms of lower $\alpha$-degree}.
\end{equation}
Thus the PBW-to-Fock determinant has a nonzero leading coefficient at each level.  For fixed $N$, exclude the finitely many zeros in $\alpha$ of the determinants through level $N$ and of the minor below.  Elsewhere the universal map is injective at every level used, so the chosen Fock minor proves generic independence on the $(c,h)$-plane.

For $n>0$, set
\begin{equation}
 \gamma_n=\alpha-Q(n+1),
 \qquad
 \delta_n=\alpha+Q(n-1).
\end{equation}
Separation into creation and annihilation parts gives
\begin{align}
 L_n&=\gamma_na_n+\frac12\sum_{r=1}^{n-1}a_ra_{n-r}+\sum_{k\ge1}a_{-k}a_{n+k},\label{eq:Ln-fock}\\
 L_{-n}&=\delta_na_{-n}+\frac12\sum_{r=1}^{n-1}a_{-r}a_{-(n-r)}+\sum_{k\ge1}a_{-(n+k)}a_k.
\label{eq:Lminusn-fock}
\end{align}
If $m_j(\lambda)$ is the multiplicity of the part $j$ in $\lambda$, the positive modes act on oscillator monomials by
\begin{align}
 L_n\ket\lambda_{\FF}
 &=\gamma_n n m_n(\lambda)\ket{\lambda-n}_{\FF}
 +\frac12\sum_{r=1}^{n-1}r(n-r)m_r(\lambda)\big(m_{n-r}(\lambda)-\delta_{r,n-r}\big)
 \ket{\lambda-r-(n-r)}_{\FF}
\notag\\
&\hspace{1.5cm}
 +\sum_{k\ge1}(n+k)m_{n+k}(\lambda)\ket{\lambda-(n+k)+k}_{\FF},
\label{eq:Fock-action-positive}
\end{align}
where a term demanding a negative multiplicity is absent.  The corresponding negative-mode action is
\begin{align}
 L_{-n}\ket\lambda_{\FF}
 &=\delta_n\ket{\lambda+n}_{\FF}
 +\frac12\sum_{r=1}^{n-1}\ket{\lambda+r+(n-r)}_{\FF}
 +\sum_{k\ge1}k m_k(\lambda)\ket{\lambda-k+(n+k)}_{\FF}.
\label{eq:Fock-action-negative}
\end{align}
The factor $1/2$ is essential.

Specialize now to $Q=0$ and put $O_n=P_NL_{-n}L_nP_N$.  In the ordered level-two Fock basis $(2),(1,1)$, direct evaluation gives
\begin{equation}
O_1=\begin{pmatrix}2&2\alpha\\2\alpha&2\alpha^2\end{pmatrix},
\qquad
O_2=\begin{pmatrix}2\alpha^2&\alpha\\\alpha&1/2\end{pmatrix}.
\end{equation}
Restrict the three columns $I,O_1,O_2$ to the matrix coordinates $((2),(2))$, $((1,1),(1,1))$, and $((2),(1,1))$.  Their determinant is
\begin{equation}
 \det\begin{pmatrix}
 1&2&2\alpha^2\\
 1&2\alpha^2&1/2\\
 0&2\alpha&\alpha
 \end{pmatrix}
 =3\alpha(2\alpha^2-1),
\end{equation}
and is not the zero polynomial in $\alpha$.

At level three, with the order $(3),(2,1),(1,1,1)$, the same calculation produces
\begin{align}
O_1&=\begin{pmatrix}6&2\alpha&0\\3\alpha&\alpha^2+4&6\alpha\\0&2\alpha&3\alpha^2\end{pmatrix},\notag\\
O_2&=\begin{pmatrix}3&2\alpha&3\\3\alpha&2\alpha^2&3\alpha\\3/2&\alpha&3/2\end{pmatrix},
\qquad
O_3=\begin{pmatrix}3\alpha^2&2\alpha&0\\3\alpha&2&0\\0&0&0\end{pmatrix}.
\end{align}
Using the three diagonal coordinates and $((3),(2,1))$, one obtains the minor
\begin{equation}
 -12\alpha(\alpha^2-1)(3\alpha^2-4),
\end{equation}
which again is not identically zero.

For $N\ge4$, let $t=\alpha^2$.  The diagonal matrix element of $O_n$ has the form
\begin{equation}
 D_n(\lambda;t)=\bra{\lambda_{\FF}^{\vee}}O_n\ket\lambda_{\FF}=t\,w_n(\lambda)+C_n(\lambda),
 \qquad w_n(\lambda)=n m_n(\lambda),
\end{equation}
with $C_n$ independent of $\alpha$.  The missing term linear in $\alpha$ is not an accidental cancellation.  Decompose
\begin{equation}
 L_n=\alpha a_n+B_n,
 \qquad
 L_{-n}=\alpha a_{-n}+B_{-n},
\end{equation}
where $B_{\pm n}$ denote the $\alpha$-independent parts of \eqref{eq:Ln-fock}--\eqref{eq:Lminusn-fock}.  Oscillator number is the length of the partition.  The factor $a_{-n}a_n$ preserves that number and contributes $\alpha^2n m_n(\lambda)$ to the diagonal; $B_n$ changes it by $-2$ or $0$, while $B_{-n}$ changes it by $+2$ or $0$.  Either cross term, $a_{-n}B_n$ or $B_{-n}a_n$, changes oscillator number by an odd integer, and hence has zero matrix element between $\ket\lambda_{\FF}$ and its coefficient dual.  The only remaining contribution is $B_{-n}B_n$, which contains no $\alpha$.

Exactly two processes in $B_{-n}B_n$ return an oscillator monomial to itself: a pair of parts summing to $n$ can be annihilated and recreated, or a part $n+k$ can be lowered to $k$ and restored.  Retaining their multiplicities gives
\begin{align}
 C_n(\lambda)
 &=\sum_{1\le r<n-r}r(n-r)m_r(\lambda)m_{n-r}(\lambda)
 \notag\\
 &\quad+\mathbf 1_{2\mid n}\,\frac14\left(\frac n2\right)^2
 m_{n/2}(\lambda)\bigl(m_{n/2}(\lambda)-1\bigr)
 \notag\\
 &\quad+\sum_{k\ge1}k(n+k)m_{n+k}(\lambda)\bigl(m_k(\lambda)+1\bigr).
\label{eq:Cn-explicit}
\end{align}
When the removed parts coincide, each ordered quadratic sum contributes its factor $1/2$, accounting for the coefficient $1/4$ in the second line.

Consider the $N$ hooks together with one additional two-row partition,
\begin{equation}
 H_1=(1^N),\quad H_j=(j,1^{N-j})\ (2\le j\le N),\quad T=(2,2,1^{N-4}).
\end{equation}
Let $\Delta_N(t)$ be the determinant with these $N+1$ partitions as rows and $I,O_1,\ldots,O_N$ as columns, and perform the row operation
\begin{equation}
 R(T)\leftarrow R(T)-2R(H_2)+R(H_1).
\end{equation}
The new row has vanishing constant entry and vanishing leading frequency vector, since $w(T)=2w(H_2)-w(H_1)$.  Substitution into \eqref{eq:Cn-explicit} leaves
\begin{equation}
 (0,E_1,E_2,\ldots,E_N)=(0,-8,2,-8,2,0,\ldots,0).
\end{equation}
The multiplicity-dependent entries required for this subtraction are
\begin{align}
&C_2(H_1)=\frac14N(N-1),\qquad C_n(H_1)=0\quad(n\ne2),\notag\\
&C_1(H_2)=2(N-1),\quad C_2(H_2)=\frac14(N-2)(N-3),\quad C_3(H_2)=2(N-2),\notag\\
&C_1(T)=4(N-3),\quad C_2(T)=\frac14(N-4)(N-5),\quad C_3(T)=4(N-4),\quad C_4(T)=2,
\end{align}
and all unlisted entries vanish.  Hence $E_n=C_n(T)-2C_n(H_2)+C_n(H_1)$ assumes the successive values $-8,2,-8,2$ for $n=1,2,3,4$.

Now form the $N\times N$ hook-frequency matrix $W_{j,n}=w_n(H_j)$.  Its determinant equals
\begin{equation}
 \det W=N N!.
\end{equation}
The first row is $(N,0,\ldots,0)$; for $j\ge2$, row $j$ has $N-j$ in the first column, $j$ in column $j$, and zeros elsewhere.  Expanding successively along columns $N,N-1,\ldots,2$ leaves $N\prod_{j=2}^Nj=NN!$.  If $A^{(m)}$ denotes the hook-row matrix with columns $1,w_1,\ldots,\widehat{w_m},\ldots,w_N$, then
\begin{equation}
 \det A^{(m)}=(-1)^{m-1}N!.
\end{equation}
The sign follows from $\sum_{n=1}^Nw_n(H_j)=N$, or equivalently from the column identity $1=N^{-1}\sum_nw_n$.  Replacing $w_m$ in $W$ by the constant column divides the determinant by $N$, leaving $N!$; moving that column to the first position contributes $(-1)^{m-1}$.

Expansion of the row-reduced determinant along the final row now isolates
\begin{equation}
 [t^{N-1}]\Delta_N(t)=\sum_{m=1}^4(-1)^{N+m}E_m\det A^{(m)}=(-1)^N12N!.
\end{equation}
For $E_m$, take $tw_n$ from all the remaining $O_n$ columns.  Each signed cofactor is $(-1)^{N-1}N!$, with $E_1+E_2+E_3+E_4=-12$; the contributions cannot cancel.  Thus $I,O_1,\ldots,O_N$ are generically independent at $Q=0$, and the minor remains nonzero in the $(c,h)$ parametrization.

\section{Full proof of the exact saturation theorem}
\label{app:saturation-proof}

Work over $\mathbb Q(c,h)$; the resulting identities specialize at every point of $\cU_N^{\rm reg}$, proving \cref{thm:saturation} on the full regular locus.

\subsection{The Shapovalov-dual Whittaker chain}

Set $\delta_j=(1^j)$ and retain the vector $D_j^\#$ defined in \eqref{eq:Dsharp-main}.  For $m\ge1$ and $\rho\vdash j-m$, contravariance gives
\begin{equation}
 G_{j-m}(L_mD_j^\#,\ket\rho_{\PBW})
 =G_j(D_j^\#,L_{-m}\ket\rho_{\PBW}).
\end{equation}
For $m\ge2$, every PBW monomial on the right contains a part of size at least two and therefore has zero $(1^j)$ coefficient; since $G_{j-m}$ is nondegenerate, $L_mD_j^\#=0$.  When $m=1$, the coefficient of $(1^j)$ in $L_{-1}\ket\rho_{\PBW}$ is one for $\rho=(1^{j-1})$ and zero otherwise.  Thus $L_1D_j^\#=D_{j-1}^\#$, which proves \eqref{eq:Whittaker-main}; iteration yields
\begin{equation}
 L_1^aD_j^\#=D_{j-a}^\#,
 \qquad 0\le a\le j,
\end{equation}
whereas any positive-mode word containing an $L_m$ with $m\ge2$ kills $D_j^\#$ once the factors to its right have carried the vector down the $L_1$ chain.

\subsection{Lower half: star-supported alternating forms}

Let $E_{\lambda\mu}$ be the PBW matrix unit with row $\lambda$ and column $\mu$.  A matrix $B$ of bilinear-form coefficients defines the functional
\begin{equation}
 \varphi_B(O)=\Tr(G_N^{-1}BO).
\label{eq:lowered-pairing-app}
\end{equation}
If $\varphi_B(O)=0$ for every $O\in\End(H_N)$, trace duality gives $G_N^{-1}B=0$, hence $B=0$.  The pairing is nondegenerate in $B$, and its trace-dual representative is $Y_B=G_N^{-1}B$.

For $N\ge3$, excise the all-ones partition $\delta_N=(1^N)$ and the hook $\eta_N=(N-1,1)$, and set
\begin{equation}
 \Lambda_N^\star=\cP(N)\setminus\{\delta_N,\eta_N\}.
\end{equation}
At $N=2$, use instead $\Lambda_2^\star=\{(2)\}$.  Each $\lambda\in\Lambda_N^\star$ then determines the alternating matrix unit
\begin{equation}
 B_\lambda=E_{\lambda,\delta_N}-E_{\delta_N,\lambda},
 \qquad
 \cB_N^\star=\Span\{B_\lambda:\lambda\in\Lambda_N^\star\}.
\end{equation}
Accordingly, $\dim\cB_N^\star=p(N)-2$ for $N\ge3$.

For arbitrary $O\in\End(H_N)$, the pairing takes the form
\begin{equation}
 \Tr(G_N^{-1}B_\lambda O)
 =\bra{\lambda^\vee}(O^\dagger-O)D_N^\#.
\label{eq:antiadjoint-app}
\end{equation}
To see this, write $\lambda^\#=G_N^{-1}\ket\lambda_{\PBW}$ and evaluate the two matrix units:
\begin{align}
 \Tr(G_N^{-1}B_\lambda O)
 &=\bra{\delta_N^\vee}O\ket{\lambda^\#}
   -\bra{\lambda^\vee}O\ket{D_N^\#}\\
 &=G_N(D_N^\#,O\lambda^\#)-\bra{\lambda^\vee}O\ket{D_N^\#}\\
 &=\bra{\lambda^\vee}(O^\dagger-O)D_N^\#.
\end{align}

For $N\ge3$, the rows surviving this restriction are indexed by
\begin{equation}
 \cD_N=\left\{\alpha\vdash d:
 2\le d\le N-2,\ \alpha\ne(1^d),\ d+\ell(\alpha)\le N\right\}.
\label{eq:Dset-app}
\end{equation}
Each such partition defines a row on the star leaves:
\begin{equation}
 r_\alpha^{(N)}(\lambda)
 =\bra{\lambda^\vee}L_{-\alpha}D_{N-|\alpha|}^\#,
 \qquad \lambda\in\Lambda_N^\star.
\label{eq:distinguished-row-app}
\end{equation}

\begin{lemma}[Star row-span]
\label{lem:star-rowspan-app}
The restrictions to $\cB_N^\star$ of all order-$N$ normal-generator pairing equations have row span
\begin{equation}
 \Span\{r_\alpha^{(N)}:\alpha\in\cD_N\}.
\end{equation}
\end{lemma}

\begin{proof}
By \cref{thm:finite-generator}, it suffices to consider $O_{\alpha,\beta}=L_{-\alpha}L_\beta$, with $|\alpha|=|\beta|=d$ and $\ell(\alpha)+\ell(\beta)\le N$.  Equation~\eqref{eq:antiadjoint-app} identifies the restricted row with the star coordinates of $(O_{\alpha,\beta}^\dagger-O_{\alpha,\beta})D_N^\#$, while the adjoint reverses both ordered blocks:
\begin{equation}
 O_{\alpha,\beta}^\dagger
 =L_{-\beta_{\ell(\beta)}}\cdots L_{-\beta_1}
  L_{\alpha_{\ell(\alpha)}}\cdots L_{\alpha_1}.
\end{equation}
Among positive blocks of weight $d$, only $L_1^d$ survives on the Whittaker chain.  If a factor $L_m$ with $m\ge2$ occurs, the factors to its right either annihilate $D_N^\#$ or lower it along that chain; in the latter case $L_mD_j^\#=0$ still applies.  The sole exception, $L_1^d$, maps $D_N^\#$ to $D_{N-d}^\#$.

Consequently both positive blocks vanish unless at least one of $\alpha,\beta$ is $(1^d)$, while self-adjointness makes the alternating row vanish when both are.

Suppose first that $\beta=(1^d)$ and $\alpha\ne(1^d)$.  The surviving term is
\begin{equation}
 (O_{\alpha,\beta}^\dagger-O_{\alpha,\beta})D_N^\#
 =-L_{-\alpha}D_{N-d}^\#.
\end{equation}
There is no admissible $\alpha$ of this type for $d=0$ or $1$.  For $2\le d\le N-2$, the length condition $d+\ell(\alpha)\le N$ is exactly the condition $\alpha\in\cD_N$, and the resulting row is $-r_\alpha^{(N)}$.  At $d=N-1$, only $\alpha=(N-1)$ remains; the vector $L_{-(N-1)}D_1^\#$ is proportional to the excised hook $(N-1,1)$, so its restriction to the star is zero.  At $d=N$, the length bound admits no nonempty $\alpha$.

In the complementary case, $\alpha=(1^d)$ and $\beta\ne(1^d)$.  PBW reordering gives
\begin{equation}
 L_{-\beta_{\ell(\beta)}}\cdots L_{-\beta_1}
 =\sum_\gamma c_{\beta\gamma}L_{-\gamma}.
\end{equation}
Since negative-mode commutators only merge parts, $c_{\beta\gamma}\ne0$ implies $\gamma\vdash d$, $\gamma\ne(1^d)$, and $\ell(\gamma)\le\ell(\beta)$.  Again $d=0,1$ cannot occur.  If $2\le d\le N-2$, then $d+\ell(\gamma)\le d+\ell(\beta)\le N$, so every summand is a distinguished row; at $d=N-1$, the only possibility is $\beta=(N-1)$ and the contribution is supported on the omitted hook, while $d=N$ contributes nothing.  Every normal-generator row therefore vanishes, is removed with the hook, or lies in the asserted span.  Conversely, for each $\alpha\in\cD_N$, the admissible generator $O_{\alpha,1^{|\alpha|}}$ produces $-r_\alpha^{(N)}$, proving equality of the spans.
\end{proof}

\begin{lemma}[Distinguished-row count]
\label{lem:row-count-app}
For $N\ge3$,
\begin{equation}
 |\cD_N|=p(N)-2-\left\lfloor\frac N2\right\rfloor.
\end{equation}
\end{lemma}

\begin{proof}
For $\alpha=(\alpha_1,\ldots,\alpha_r)\in\cD_N$ of weight $d$, define
\begin{equation}
 F_N(\alpha)=(\alpha_1+1,\ldots,\alpha_r+1,1^{N-d-r}).
\end{equation}
Deleting the unit parts of $F_N(\alpha)$ and decrementing the rest recovers $\alpha$, so $F_N$ is injective.  Conversely, a partition $\lambda\vdash N$ with $u$ units yields $|\alpha|+\ell(\alpha)=N-u\le N$.  The output lies outside $\cD_N$ only for $\lambda=(1^N)$, $(N)$, or $(2^a,1^{N-2a})$, $1\le a\le\lfloor N/2\rfloor$; these yield $\varnothing$, $(N-1)$, or $(1^a)$, respectively.  Removing these $2+\lfloor N/2\rfloor$ exceptions from $p(N)$ proves the formula.
\end{proof}

Let $R_N:\cB_N^\star\to A_N(N)^*$ be the restriction of the pairing, $R_N(B)(O)=\Tr(G_N^{-1}BO)$.  For $N\ge3$, the two lemmas and rank--nullity yield
\begin{align}
 \dim\ker R_N
 &\ge (p(N)-2)-\left(p(N)-2-\left\lfloor\frac N2\right\rfloor\right)\\
 &=\left\lfloor\frac N2\right\rfloor.
\end{align}
For $N=2$, the star is one-dimensional and $A_2(2)$ is spanned by $I_2$, $P_2L_{-1}L_1P_2$, and $P_2L_{-2}L_2P_2$.  All three are self-adjoint, so the anti-adjoint pairing vanishes on the entire star and gives the same lower bound, namely one.  If $B\in\ker R_N$, then $Y_B=G_N^{-1}B$ satisfies $\Tr(Y_BO)=0$ for every $O\in A_N(N)$.  Since multiplication by $G_N^{-1}$ is injective, linearly independent nonzero forms in the kernel give linearly independent nonzero elements of 
%$\Ann_{\Tr}(A_N(N))$
$I_k(N)$.  This proves \eqref{eq:saturation-ann-bound} and \eqref{eq:saturation-lower}.

\subsection{Upper half: cyclicity and transition density}

The transition spaces are those of \eqref{eq:transition-space}; contravariance implies $\cT_k(r,s)^\dagger=\cT_k(s,r)$.

Introduce the index set
\begin{equation}
 \cI_N=\{\alpha:|\alpha|+\ell(\alpha)\le N\}
\end{equation}
including the empty partition, and set
\begin{equation}
 v_\alpha^{(N)}=L_{-\alpha}D_{N-|\alpha|}^\#
 =L_{-\alpha}L_1^{|\alpha|}D_N^\#.
\end{equation}
The word on the right has at most $N$ factors.  Identify $\cI_N$ with the partitions of $N$ through the map
\begin{equation}
 \Phi_N(\alpha)=(\alpha_1+1,\ldots,\alpha_{\ell(\alpha)}+1,
 1^{N-|\alpha|-\ell(\alpha)}).
\label{eq:Phi-app}
\end{equation}
Deleting the unit parts and subtracting one from every remaining part gives the inverse.  Thus $\Phi_N$ is a bijection from $\cI_N$ to the partitions of $N$.

\begin{lemma}[Whittaker cyclicity]
\label{lem:cyclicity-app}
The vectors $\{v_\alpha^{(N)}:\alpha\in\cI_N\}$ form a basis of $H_N$.  Hence
\begin{equation}
 A_N(N)D_N^\#=H_N.
\label{eq:cyclicity-app}
\end{equation}
\end{lemma}

\begin{proof}
Consider the square matrix of pairings
\begin{equation}
 C_{\alpha\gamma}=G_N(v_\alpha^{(N)},\ket{\Phi_N(\gamma)}_{\PBW}).
\end{equation}
Move the negative block across the Shapovalov form and use the coefficient-dual definition of $D_{N-|\alpha|}^\#$ to obtain
\begin{equation}
 C_{\alpha\gamma}
 =\operatorname{Coeff}_{(1^{N-|\alpha|})}
 \left(L_{\alpha_{\ell(\alpha)}}\cdots L_{\alpha_1}
 \ket{\Phi_N(\gamma)}_{\PBW}\right).
\end{equation}
Positive Virasoro modes cannot increase PBW length.  The initial and target partitions have lengths $N-|\gamma|$ and $N-|\alpha|$; hence $C_{\alpha\gamma}=0$ whenever $|\alpha|<|\gamma|$.

Suppose $|\alpha|=|\gamma|$, so that the initial and final PBW lengths agree.  Every contributing commutator branch must then preserve length at each step, because a lost PBW factor cannot be restored by later positive modes.  The length-preserving part of the $L_j$ action replaces a part $q>j$ by $q-j$ through
\begin{equation}
 [L_j,L_{-q}]=(j+q)L_{-(q-j)}.
\end{equation}
Intermediate reductions $q\mapsto q-j>1$ are allowed; terms with $q=j$, central contractions, positive remainders, joining, or splitting lower PBW length and cannot contribute.

The partition $\Phi_N(\gamma)$ has $\ell(\gamma)$ non-unit parts, each of which must be struck, since an action elsewhere cannot remove it without lowering the length.  Nonvanishing thus requires $\ell(\alpha)\ge\ell(\gamma)$.  If equality holds, the number of positive modes equals the number of non-unit parts and each is acted upon exactly once; no further action is available to complete an intermediate reduction.  A part $q$ must therefore be sent directly to $1$ by a mode of label $q-1$.  As the non-unit parts are $\gamma_i+1$, the mode labels form exactly the multiset of parts of $\gamma$, forcing $\alpha=\gamma$.

On the diagonal, modes of label $j$ may be paired with the parts $j+1$ in $m_j(\alpha)!$ ways, and each pairing contributes $2j+1$.  Thus
\begin{equation}
 C_{\alpha\alpha}=\prod_{j\ge1}m_j(\alpha)!\,(2j+1)^{m_j(\alpha)}\ne0,
\end{equation}
because $[L_j,L_{-(j+1)}]=(2j+1)L_{-1}$.  Order $\cI_N$ by increasing weight and, within each weight, by increasing length.  With this ordering, $C$ is triangular with the displayed nonzero diagonal, and the bijection \eqref{eq:Phi-app} makes it a $p(N)\times p(N)$ matrix.  Hence the vectors $v_\alpha^{(N)}$ form a basis.
\end{proof}

Set
\begin{equation}
 q_N(v)=G_N(D_N^\#,v),
 \qquad
 \cS_N=\ker q_N
 =\Span\{\ket\lambda_{\PBW}:\lambda\ne(1^N)\}.
\end{equation}
Every partition other than $(1^N)$ has a part $m\ge2$, so
\begin{equation}
 \cS_N=\sum_{m=2}^{N}L_{-m}H_{N-m}.
\label{eq:Ssum-app}
\end{equation}
For a transition $T:H_s\to H_r$, contravariance reads
\begin{equation}
 q_r(Tv)=G_s(T^\dagger D_r^\#,v).
\label{eq:q-adjoint-app}
\end{equation}

\begin{proof}[Proof of \cref{lem:transition-density}]
Induct on $M=\max(r,s)$.  Since $H_0$ is a line and creation words span $H_M$, transitions to or from $H_0$ are available; diagonal cases with $M\le1$ are also one-dimensional.

For unequal levels take $r=M>s\ge1$.  Because the sum in \eqref{eq:Ssum-app} need not be direct, lift maps into $\cS_r$ through
\begin{equation}
 F_r:\bigoplus_{m=2}^{r}H_{r-m}\longrightarrow\cS_r,
 \qquad
 F_r((v_m)_m)=\sum_{m=2}^{r}L_{-m}v_m.
\label{eq:transition-lift-map}
\end{equation}
Given $T:H_s\to\cS_r$, select a lift
\begin{equation}
 \widetilde T:H_s\longrightarrow\bigoplus_{m=2}^{r}H_{r-m},
 \qquad F_r\widetilde T=T;
\end{equation}
using a right inverse on $T$ and $F_r$.  Write $\widetilde T=(\widetilde T_m)_m$.  Induction realizes each $\widetilde T_m:H_s\to H_{r-m}$ within $M-1$ modes.  When $r-m=s$, $s+1\le r-1=M-1$.  Left multiplication by $L_{-m}$ gives
\begin{equation}
 T=\sum_{m=2}^{r}L_{-m}\widetilde T_m\in\cT_M(r,s),
 \qquad
 \operatorname{Hom}(H_s,\cS_r)\subseteq\cT_M(r,s).
\end{equation}

For the one-dimensional quotient $H_r/\cS_r$, fix $u\in H_s$.  Cyclicity gives $A_u\in A_s(s)$ with $A_uD_s^\#=u$; define
\begin{equation}
 Y_u=A_uL_1^{r-s}\in\cT_r(s,r),
 \qquad
 T_u=Y_u^\dagger\in\cT_r(r,s).
\end{equation}
Then $Y_uD_r^\#=u$, and \eqref{eq:q-adjoint-app} gives $q_r(T_uv)=G_s(u,v)$.  Nondegeneracy of $G_s$ makes the maps $v\mapsto[T_uv]$ span $\operatorname{Hom}(H_s,H_r/\cS_r)$.  Together with $\operatorname{Hom}(H_s,\cS_r)\subseteq\cT_M(r,s)$, this proves $\cT_M(r,s)=\operatorname{Hom}(H_s,H_r)$; adjunction gives the case $s>r$.

For $r=s=N$ and $T:H_N\to\cS_N$, choose $F_N\widetilde T=T$.  The unequal-level case places each $\widetilde T_m$ in $\cT_N(N-m,N)$.  Left multiplication by $L_{-m}$ costs one further mode; summing the terms gives $\operatorname{Hom}(H_N,\cS_N)\subseteq\cT_{N+1}(N,N)=A_{N+1}(N)$.  Cyclicity also gives, for each $u\in H_N$, an $A_u\in A_N(N)$ with $A_uD_N^\#=u$.  Equation~\eqref{eq:q-adjoint-app} gives $q_N(A_u^\dagger v)=G_N(u,v)$, so $A_u^\dagger$ realizes all maps from $H_N$ to $H_N/\cS_N$.  With maps into $\cS_N$, these fill $\End(H_N)$.
\end{proof}

The second identity in \cref{lem:transition-density} is \eqref{eq:saturation-upper}.  Failure at order $N$ and monotonicity then give $k_{\rm sat}(N)=N+1$.

\end{document}